\documentclass[a4paper,twocolumn,11pt,unpublished]{quantumarticle}
\pdfoutput=1

\makeatletter
\long\def\addtocontents#1#2{%
  \protected@write\@auxout
    {\let\label\@gobble \let\index\@gobble \let\glossary\@gobble \let\footnote\@gobble}%
    {\string\@writefile{#1}{#2}}}
\makeatother

\usepackage{amsmath,amssymb,amsthm,mathtools}
\usepackage{hyperref}
\usepackage{cleveref}
\usepackage{braket}
\usepackage{verbatim}

\numberwithin{equation}{section}

\usepackage[
  backend=biber,
  style=numeric-comp
]{biblatex}

\AtEveryBibitem{%
  \clearfield{note}%
  \clearfield{annotation}%
  \clearfield{addendum}%
}

\newtheorem{problem}{Problem}
\newtheorem{proposition}{Proposition}
\newtheorem{definition}{Definition}
\newtheorem*{question*}{Question}
\newtheorem*{result*}{Result}
\newtheorem*{definition*}{Definition}

\title{How to make quantum cheese: efficient geometry oracles for exponentially many pseudorandom microstructures}
\author{Alice Barthe}
\affiliation{PsiQuantum, 700 Hansen Way, Palo Alto, CA 94304, USA}
\date{}

\begin{document}

\begin{abstract}
Quantum algorithms for partial differential equations often assume oracle access to the geometry of the material being simulated.
As explicitly specifying a geometry pixel by pixel on an exponentially large grid automatically yields an exponential cost, a central question is when such oracles can effectively be implemented by polynomially sized quantum circuits.
We give explicit polynomial-size circuits for a family of geometries with exponentially many pseudorandom microscopic features.
The geometry is generated by a pseudorandom procedure and a convolution with a local kernel to shape the texture.
This proves to be a versatile framework, that includes materials with fixed-shape holes, spatially correlated textures, and activation--inhibition patterns.
For a fixed kernel with bounded integer coefficients, the oracle cost is polynomial in the number of qubits, even when the number of microfeatures is exponential.
We implement the circuit constructions in PsiQuantum Construct and verify their behaviour through numerical simulation.
We also discuss why efficient access does not follow from an arbitrary description of the material, by proving worst-case query lower bounds for amplitude-encoded formation fields and databases of holes.
\end{abstract}

\section{Introduction}
Many quantum algorithms for linear systems, Hamiltonian simulation, and differential-equation subroutines are formulated in an input-access model.
Their complexity is typically measured in terms of queries to sparse-access oracles, or more generally to block-encodings of the relevant operator~\cite{gilyén2019quantum-d2b}. 
This abstraction separates the algorithmic cost of processing an input from the cost of coherently providing the input for a specific problem instance.
A central practical question is therefore when such an abstract access model can be implemented by a polynomial-size quantum circuit.

This question is particularly important for quantum algorithms for partial differential equations. 
After discretization on a grid, such algorithms can, under suitable assumptions, achieve complexity scaling sublinearly in the number of grid points, improving over classical methods whose cost is typically at least linear in the grid size~\cite{childs2021high-precision-36a,montanaro2016quantum-e06,harrow2009quantum-373}. 
This potential speed-up is often used to increase the spatial resolution of the simulation. 

In this work, we use the power of the quantum computer to simulate a larger volume of material at a fixed microscopic resolution. 
The resolution is set to a particular scale required to resolve the physics of interest, and adding a coordinate qubit doubles the number of sites along that coordinate.
In the case of a textured material, the number of pores or micro-features grows with the volume, and therefore exponentially with the number of qubits. 
An explicit list of all those features can become prohibitively long, even though a quantum register uses only $n$ qubits to address a grid with $2^n$ sites. 
We therefore focus on materials whose geometry is specified by a generative rule rather than by an explicit list of all microscopic features.
Examples include porous media~\cite{hyman2015statistical-249,hyman2014stochastic-b66}, as well as morphologies arising from dynamical formation processes.

This motivates the question that will guide the rest of this paper:
\begin{question*}
    Can exponentially many geometrical features be encoded by a polynomial-size quantum circuit?
\end{question*}
To that effect we introduce the notion of geometry oracle, and the following associated problem. 
\begin{problem}
    \label{prob:geomoracle}
    Let $g : \mathcal{D} \to \{0,1\}$ describe a material on a discretized domain $\mathcal{D}$ with $N=2^n$ sites. We use $g(x)=1$ for void and $g(x)=0$ for material.
    The goal is to compile a polynomial-size quantum circuit for the corresponding geometry oracle
    $$
    G\lvert x\rangle \lvert b\rangle = \lvert x\rangle \lvert b\oplus g(x)\rangle.
    $$
\end{problem}
The need for such oracles appears, for example, in quantum algorithms for fluid-dynamics where one must distinguish wall nodes from fluid nodes. 
Explicit examples are found in Equation (26) in~\cite{jennings2025simulating-a85} or in Equation (96) in~\cite{schleich2025arbitrary-4d0}.
Microstructure-resolved simulations could also be used to estimate effective mechanical properties of materials, such as the elastic moduli of porous foams and cellular solids~\cite{carolis2024effect-186,roberts2000elastic-50a}.

Efficient geometry access is not easy for all geometries, we provide two scenarios with unfavourable complexities.
Converting an amplitude-encoded field on $N$ sites into coherent sign access requires $\Omega(\sqrt N)$ state-preparation queries in the worst case, even with controlled access to the preparation unitary and its inverse (\Cref{thm:equadiffmat_hard}).
Evaluating the union of $M$ arbitrary holes supplied by a quantum database requires $\Omega(\sqrt M)$ database queries in the worst case (\Cref{thm:arbitrarydisks_hard}).
These bounds concern the specified black-box access models, not every implementation of a geometry oracle.
The proofs are in Appendix~\ref{sec:lowerbounds}.
They motivate the efficiently evaluable generative rules used below.

Our construction describes the material by a rule rather than a list of features.
A keyed pseudorandom permutation of the grid addresses creates the initial randomness of the materials.
Thresholding the permuted addresses gives a binary centre field. A fixed local convolution rule turns this field into a more realistic texture.
For one queried coordinate, the circuit only evaluates the relevant neighbourhood, it never needs to search an exhaustive list of microstructures.
The number of microstructures and the total number of sites may grow, while the local rule stays fixed.
This corresponds to increasing the material volume at fixed microscopic resolution.
\begin{result*}
    There is a family of materials with local microstructures as defined in \Cref{def:loc_material}, whose geometry oracles on exponentially increasing volumes, and therefore exponentially many microstructures, can be implemented by polynomial-size quantum circuits.
\end{result*}

We first explain the classical geometry in \Cref{sec:geometry}, then its reversible implementation as a quantum circuit in \Cref{sec:circuits}.
\Cref{sec:implementation} describes the Construct implementation and resource scaling.

\section{Classical generative rule}
\label{sec:geometry}
In this section we describe the generative rule for the geometry.
We first explain how we can create materials with randomly placed holes of fixed shape.
Then we explain how this generalises to a more general family of materials that are random binary fields convolved with a kernel.

\subsection{Pseudorandom permutation}
\begin{figure}[h]
% \begin{figure}
    \centering
    \includegraphics[width=\linewidth]{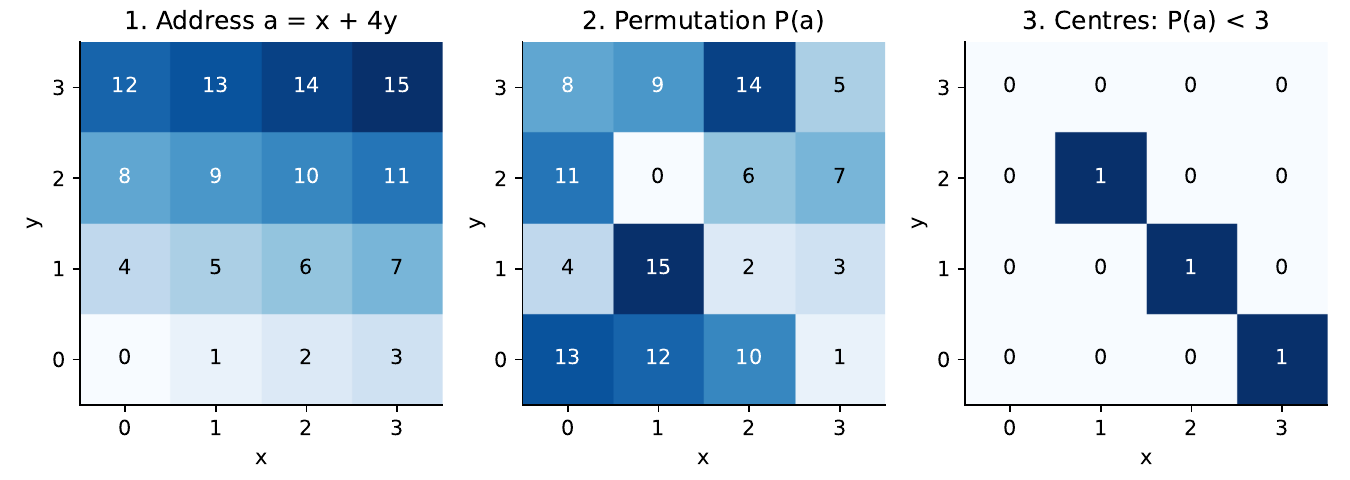}
    \caption{An illustrative example of the three steps of the construction of centres on a $4\times4$ grid. 1. Coordinates receive unique addresses \texttt{a=x+4y}. 2. The addresses are shuffled pseudorandomly. 3. Values below \texttt{M=3} select centres \texttt{(1,2),(2,1),(3,0)}.}
    \label{fig:shuffle}
\end{figure}
To generate the material we need a source of randomness that is uniformly distributed across the material domain.
In order to access this, we implement a pseudorandom permutation of all the coordinate addresses in the domain.
% We now choose an implementation of the permutation used in \Cref{eq:def_random_centers}.
Any reversible family of pseudorandom permutations with polynomial circuit size would be suitable, for example Feistel-type constructions~\cite{luby1988how-668}.
For concreteness, we choose the symmetric cipher SPECK~\cite{cryptoeprint:2013404}.

The SPECK cipher encrypts an $n$-bit long message $x$, with $n$ necessarily even.
First we split the message into two $n/2$-bit registers, traditionally called left and right.
One round of encryption is composed of rolling operations, modulo integer additions and XOR operations. 
Rolling a bit string by an integer $i$ means shifting every bit $i$ places around the string, wrapping bits that pass one end back to the other.
A SPECK cipher has two fixed rotation parameters, integers $u$ and $v$.
The round includes the application of a classical key $k$ and is detailed in \Cref{subsec:prp}.

Such rounds are applied successively with changing keys.
The round keys are a classical sequence $k=(k_1,\ldots,k_m)$, with $n/2$ bits per round.
The chosen rotations, number of rounds, and keys are part of the geometry description.

\subsection{Creating hole centres}
\label{subsec:rndcenters}
We describe how to generate the centres of the holes in the material.
Write the grid as $\mathcal{D}=\prod_{j=1}^d\mathbb{Z}_{2^{n_j}}$, with $n=\sum_j n_j$ and $N=2^n$ sites.
Concatenating the coordinate registers gives each site a unique $n$-bit address.
We use the same symbol $x$ for a coordinate and its address when the meaning is clear.
Choose a permutation $P^{(k)}$ of these addresses and an integer $0\le M\le N$.
Define the centre field
\begin{equation}
    c_{k,M}(x)=\mathbf{1}[P^{(k)}(x)<M].
    \label{eq:def_random_centers}
\end{equation}
Because the permutation is bijective, exactly $M$ sites are selected. 
% The endpoint cases $M=0$ and $M=N$ give constant fields and need no permutation or comparison circuit.
A pseudorandom permutation provides spatial disorder; changing the key changes the material realization.
This is a family of pseudorandom centre sets, not an assertion of exact uniform sampling over all sets of size $M$.
We give an example on a $4\times4$ grid in \Cref{fig:shuffle}.

\subsection{Fixed-shape holes}
\label{subsec:fixedshape}
A hole shape is a finite set of offsets
\begin{equation}
    \Delta=\{\delta_t\}_{t=1}^T\subset\mathbb{Z}^d.
\end{equation}
A site $x$ belongs to a hole if it is an offset away from a centre:
\begin{equation}
    h_{\Delta,k,M}(x)=\bigvee_{t=1}^T c_{k,M}(x-\delta_t).
    \label{eq:def_shape_map}
\end{equation}
The minus sign follows from $x=\text{centre}+\delta_t$.
For example, use the two-pixel horizontal stencil $\Delta=\{(0,0),(1,0)\}$ with the centres in \Cref{fig:shuffle}. At $(0,0)$, the centre tests query $(0,0)$ and $(3,0)$.
Their permuted addresses are $13$ and $1$, so the tests return $0$ and $1$. The queried pixel is void.

In practice this OR is implemented as follows, for each offset, increment the accumulator when the corresponding centre test returns 1. 
The number in the accumulator corresponds to the number of holes overlapping at the queried coordinate, the queried site is void if the final count is at least one.

\begin{figure}[!h]
    \centering
    \includegraphics[width=0.5\linewidth]{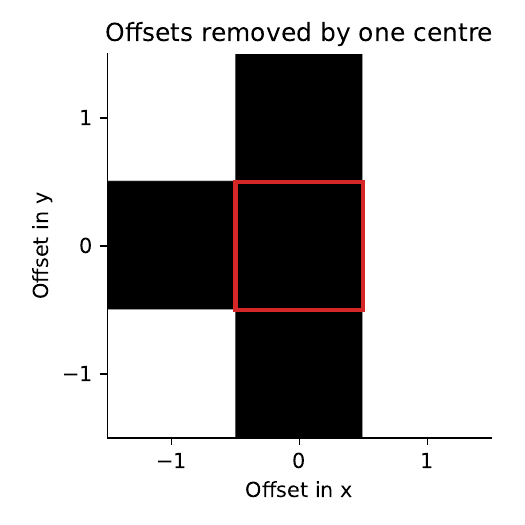}
    \caption{A four-pixel hole shape. The central pixel outlined in red has offset $(0,0)$ and the black pixels specify $\Delta=\{(0,0),(-1,0),(0,-1),(0,1)\}$. This stencil is used in \Cref{fig:quantumcheese}.}
    \label{fig:diamond}
\end{figure}

\subsection{Local convolution and thresholding}
\label{subsec:convkernel}
This local hole offset rule combined with a final OR rule generalizes.
The initial centre map can be seen as a random binary field.
This binary field can be shaped by the convolution with a kernel as a weighted sum.
This integer-valued field can finally be compared to a threshold to yield the geometry as a binary map.
Let a kernel contain $R$ non-zero integer coefficients at offsets $\eta_r$:
\begin{equation}
    \Lambda=\{(\eta_r,\alpha_r)\}_{r=1}^R,\qquad\alpha_r\in\mathbb{Z}.
\end{equation}
Each centre test contributes its offset's weight if it returns $1$.
Positive weights reinforce the score; negative weights reduce it.
We give a formal definition of such textured material below, and illustrate this process in \Cref{fig:textures}.

\begin{definition}[Locally textured material]
    A material is locally textured if, on a discretized domain $\mathcal{D}$ with $N=2^n$ sites, its geometry can be generated by the following stochastic rule.
    Sample a seed $k$ once to specify a random binary field $c_{k,M}(x)$ with exactly $M$ entries equal to $1$, where $0\le M\le N$.
    For a fixed seed, the field is deterministic and efficiently computable at each queried coordinate.
    For a finite integer kernel $\Lambda$ and an integer threshold $\tau$, both independent of $n$, define
    \begin{equation}
        g(x)=\mathbf{1}[(c_{k,M}*\Lambda)(x)\ge\tau],
        \label{eq:def_threshold_geometry}
    \end{equation}
    where $g:\mathcal{D}\to\{0:\text{material},1:\text{void}\}$ and $*$ denotes convolution with periodic boundaries.
    \label{def:loc_material}
\end{definition}

Fixed-shape holes are the special case
\begin{equation}
    \Lambda=\{(\delta,1):\delta\in\Delta\},\qquad\tau=1.
\end{equation}
The sum then counts how many centres cover a site, and thresholding at one gives the OR.
\begin{figure*}[t]
    \centering
    \includegraphics[width=0.76\textwidth]{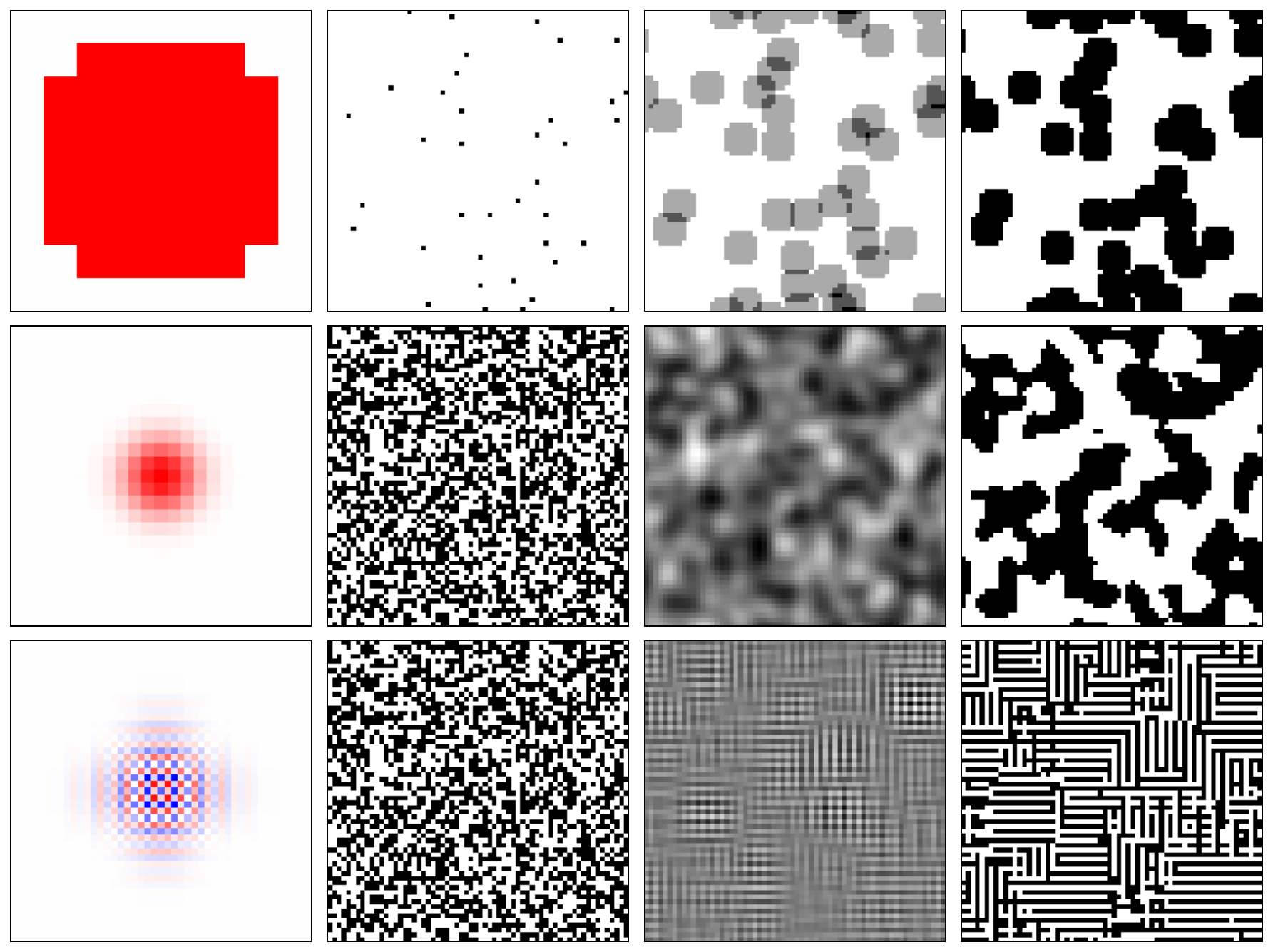}
    \caption{Several geometries for which we provide a recipe to explicitly compile a polynomially-sized quantum geometry oracle. Reading the plot from left to right are different steps in the process of generating the geometry. Each row corresponds to a different model of material. The \textbf{first column} shows the kernel $\Lambda$, where positive values are red, negative values are blue and zeros are white. The \textbf{second column} shows the pseudorandom binary field $c_{k,M}$, where black is 1 and white is 0. The \textbf{third column} shows the greyscale image resulting from the convolution of the kernel with the binary field. The \textbf{last column} shows the greyscale image thresholded, aka the geometry function $g$. In the \textbf{top row}, we show a model with pseudorandom disks uniformly distributed. In the \textbf{middle row} we use a Gaussian smoothing kernel, another frequent model for porous materials, ceramics, minerals or foams. In the \textbf{bottom row}, the kernel yields two-scale activation-inhibition textures analogous to competing-interaction patterns.}
    \label{fig:textures}
\end{figure*}

\subsection{Examples of kernels and textures}
We show that the locally textured definition yields a broad family of materials beyond fixed-shape holes.
We illustrate this diversity in \Cref{fig:textures}, and provide intuition for these textures below.

\textbf{Random circular holes.}
Circular holes are obtained by choosing a disk-shaped kernel with unit weights and threshold $\tau=1$.
Combining disks of different radii gives models of porous materials such as those studied in~\cite{carolis2024effect-186,roberts2000elastic-50a}.

\textbf{Gaussian smoothing.}
A positive kernel approximating a truncated Gaussian smooths the binary centre field.
Nearby sites receive contributions from overlapping neighbourhoods, creating a higher correlation length than the initial speckle noise.
Thresholded correlated random fields are used to model porous media, ceramics, foams, and mineral microstructures~\cite{hyman2014stochastic-b66,roberts2001elastic-145,roberts2000elastic-50a,teichmann2021modeling-2e8}.

\textbf{Activation--inhibition.}
Positive short-range weights and negative longer-range weights produce competing local contributions.
For a field $u^{(t)}$, let $S^{(t)}$ and $L^{(t)}$ be its averages on a small neighbourhood and a disjoint surrounding neighbourhood.
An example update is
\begin{equation}
    u^{(t+1)}=aS^{(t)}-bL^{(t)}+cu^{(t)},\qquad a,b,c>0.
    \label{eq:updaterule}
\end{equation}
This is a linear convolution. Repeating it $p$ times gives a kernel equal to its $p$-fold convolution, followed by a single final threshold.
The example uses $a=1.2$, $b=1.0$, $c=0.2$, $p=10$, and threshold $0.5$, with averages over the $3\times3$ and $5\times5$ outer rings.
The $p$-fold convolution yields rational coefficients, which fit the integer-kernel convention by multiplying by an integer. 
The textures resemble competing-interaction patterns in reaction--diffusion systems, precipitation, colloids, and polymers~\cite{kondo2010reaction-diffusion-fd3,nabika2020pattern-0f6,liu2019colloidal-7ba,ohta1986equilibrium-5d6,carretas-talamante2023non-equilibrium-1d5}.
This is a morphological model, not a simulation of those formation dynamics.

\section{Reversible geometry oracles}
\label{sec:circuits}
The classical rule fixes what the circuit must compute.
In this section, we build a reversible quantum circuit that preserves the queried coordinate and returns temporary registers to zero.
We use compute--copy--uncompute architectures: evaluate a function, XOR its bit into the output, and reverse the intermediate computation.

\subsection{Circuit primitives}
\label{subsec:primitives}
In this subsection we collect the list of reversible blocks that will be used throughout the rest of the construction, and fix the associated notation.
For the same reversible operation, several implementations exist with different costs in terms of number of gates, depth and number of ancillary qubits, clean or dirty.
We do not claim originality for these primitives, we include them for completeness.

Our notation is as follows: For a circuit, the subscripts indicate the registers on which it acts.
When the action is asymmetric, for example one register is modified in place while the other is left unchanged, we indicate this with an arrow in the subscript.
Numerical parameters are written as superscripts.
All unsigned arithmetic on a $q$-qubit register is modulo $2^q$. 

\textbf{Register XOR.}
For two $q$-qubit registers $a$ and $b$, we write
\begin{equation}
    C_{a\to b}\ket{x}_a\ket{y}_b = \ket{x}_a\ket{y\oplus x}_b.
\end{equation}
This is implemented by $q$ CNOT gates, with a constant depth of 1 and is used inside a cipher round.

\textbf{Quantum-quantum addition.}
For two $q$-qubit registers $a$ and $b$, we write
\begin{equation}
    A_{a\to b}\ket{x}_a\ket{y}_b = \ket{x}_a\ket{x+y}_b.
\end{equation}
There are different implementations of such an in-place adder, some with lower gate count~\cite{gidney2018halving-981}, some in-between~\cite{cuccaro2004new-62d} and one with $O(\log^2 q)$ depth and no ancillary qubits~\cite{remaud2025ancilla-free-eb8}.

\textbf{Classical-quantum addition.}
For a $q$-qubit register $x$ and a fixed integer $\delta$, we write
\begin{equation}
    A_{x}^{(\delta)}\ket{z}_x = \ket{z+\delta}_x.
\end{equation}
Several implementations are available, including one with a constant number of ancillary qubits~\cite{gidney2025classical-quantum-cba} and one~\cite{vandaele2026asymptotically-8cd} with $O(\log^2 q)$ depth and one dirty ancillary qubit.
When several coordinate registers are present, for example $\mathbf{x}=(x_1,\dots,x_d)$ and $\delta=(\delta_1,\dots,\delta_d)$, we write
\begin{equation}
    A_{\mathbf{x}}^{(\delta)} = \prod_{j=1}^d A_{x_j}^{(\delta_j)},
\end{equation}
acting component-wise.
For an addition controlled by the one-qubit control register $c$, we write
\begin{equation}
    A_{c\to x}^{(\delta)}\ket{b}_c\ket{z}_x
    =
    \ket{b}_c\ket{z + b \delta}_x.
\end{equation}
In words, the constant $\delta$ is added to $x$ if and only if the control qubit $c$ is equal to $1$.

\textbf{Comparison with a fixed threshold.}
For a $q$-qubit register $x$, a one-qubit flag register $f$, and a fixed integer $M$, we write
\begin{equation}
    L_{x\to f}^{(M)}\ket{z}_x\ket{b}_f = \ket{z}_x\ket{b\oplus\mathbf{1}[z<M]}_f.
\end{equation}
Here the flag is flipped if and only if the integer encoded by $x$ is strictly less than $M$.
The circuit construction permits any reversible constant comparator that restores its ancillary qubits, including the conditionally clean construction~\cite{khattar2025rise-320} and the asymptotically optimal construction~\cite{vandaele2026asymptotically-8cd}, with $O(\log q)$ depth and one dirty ancillary qubit.

\textbf{Bitwise XOR with a fixed constant.}
For a $q$-qubit register $x$ and a fixed bit-string $k\in\{0,1\}^q$, we write
\begin{equation}
    X_x^{(k)}\ket{z}_x = \ket{z\oplus k}_x.
\end{equation}
This is implemented by applying an $X$ gate on the $j$-th qubit whenever the $j$-th bit of $k$ is equal to $1$.

\textbf{Cyclic bit rotations.}
For a $q$-qubit register $x$ and an integer $u$, we write $R_x^{(u)}$ for the cyclic rotation of the bits of $x$ by $u$ positions.
For example
\begin{equation}
    b_0b_1b_2b_3b_4 \xrightarrow[R^{(-2)}]{} b_2b_3b_4b_0b_1.
\end{equation}
These operations are free in the logical cost model because they re-index wires. 
% An implementation emitting explicit SWAP gates must count those gates separately; hardware routing is outside this model.

\subsection{Pseudorandom permutations}
\label{subsec:prp}
We implement SPECK following the reversible construction in~\cite{anand2020evaluation-36a}, and write the left and right registers as $x=(l,r)$.
Using the primitives introduced in \Cref{subsec:primitives}, one round of the permutation is the circuit
\begin{equation}
    S_{l,r}^{(k;u,v)}
    =
    C_{l\to r}\,
    R_r^{(v)}\,
    X_l^{(k)}\,
    A_{r\to l}\,
    R_l^{(-u)}.
    \label{eq:speck_round}
\end{equation}

\begin{figure}
    \centering
    \includegraphics[width=1\linewidth]{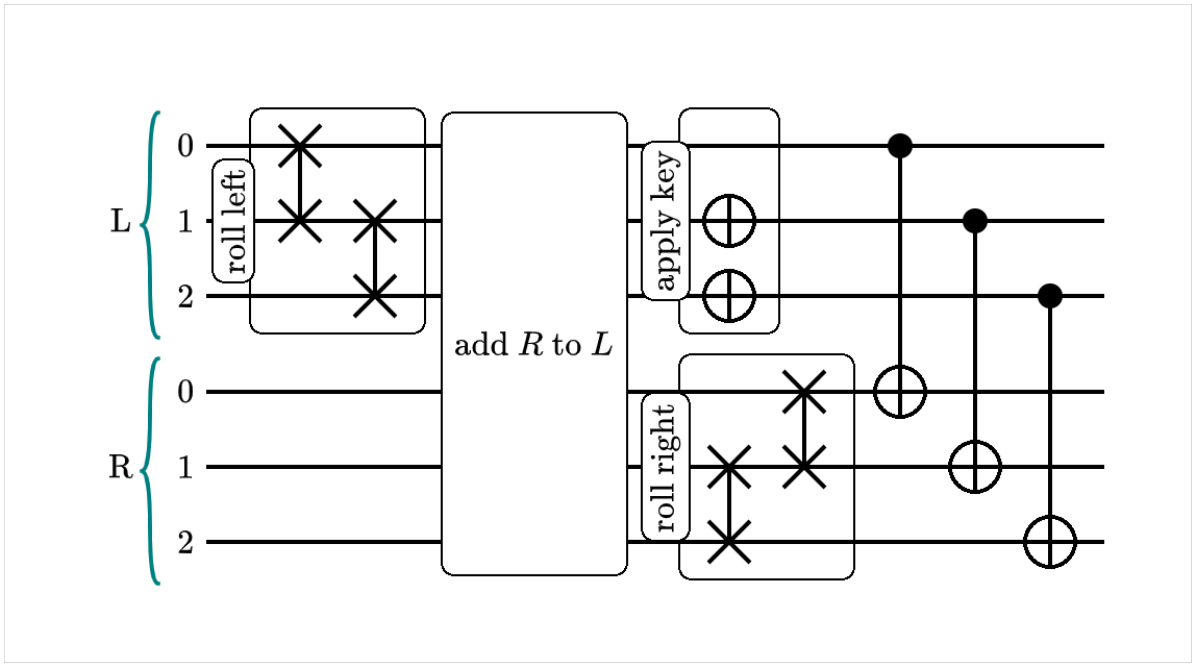}
    \caption{Circuit Designer rendering of the quantum circuit for one round of SPECK cipher applied to the left register \texttt{L} and the right register \texttt{R}. Sequentially (1) roll left \texttt{L} with \texttt{u=1} (2) integer addition of right register onto the left register \texttt{L<-L+R} (3) roll right \texttt{R} with \texttt{v=1} (4) apply the key \texttt{b110} to \texttt{L} (5) XOR the left register onto the right register \texttt{R<-R XOR L}.}
    \label{fig:circ_one_round}
\end{figure}

The circuit for one round is shown in \Cref{fig:circ_one_round}.
The full permutation is obtained by composing $m$ such rounds.
If $\mathbf{k}=(k_1,\dots,k_m)$ is the sequence of round keys, we write
\begin{equation}
    P_x^{(\mathbf{k};u,v)}
    =
    S_{l,r}^{(k_m;u,v)}
    \cdots
    S_{l,r}^{(k_1;u,v)}.
    \label{eq:full_prp}
\end{equation}
which we shorten to $P_x^{(k)}$, with $k$ denoting the full set of subkeys, also called key schedule.

Appendix~\ref{app:randomness} examines the effect of the number of rounds on the quality of the mixing.
Taking $m=\Theta(n)$ retains polynomial circuit size and gives empirically satisfactory levels of randomness for the tested parameters.
Those tests are empirical diagnostics, not a proof of uniform sampling or cryptographic security, which is unnecessary in our case.
More rounds yield a higher quality of mixing but increase the circuit cost.

\subsection{Shifted centre queries}
\label{subsec:centre_circuit}
We subtract the required offset $\delta$ from the coordinate, apply the permutation, and compare the permuted address to $M$.
\begin{equation}
    \label{eq:intro_center_oracle}
    Q_{x\to f}^{(k,M,\delta)}= L_{x\to f}^{(M)}P_x^{(k)} A_{x}^{(-\delta)}.
\end{equation}
For an arbitrary output bit $b$, the state evolves as
\begin{align}
    \ket{x}\ket{b}&\longmapsto \ket{x-\delta}\ket{b}\\
    &\longmapsto\ket{P^{(k)}(x-\delta)}\ket{b}\notag\\
    &\longmapsto\ket{P^{(k)}(x-\delta)}\ket{b\oplus c_{k,M}(x-\delta)}\notag
\end{align}
The comparator XORs its result into the flag. The coordinate remains shifted and permuted until $Q^\dagger$ is applied. 
We write this query block as $Q^{(k,M,\delta)}$ and illustrate it in \Cref{fig:circ_centres}.
\begin{figure}
    \centering
    \includegraphics[width=0.7\linewidth]{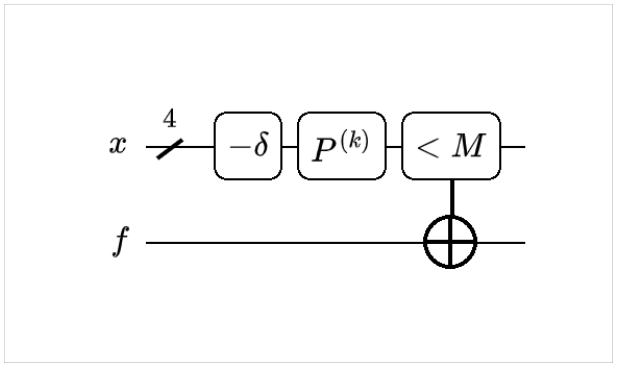}
    \caption{Shifted centre query $Q^{(k,M,\delta)}$: subtract $\delta$, apply the $m$-round permutation $P^{(k)}$, with a single round implemented as in  \Cref{fig:circ_one_round}, and XOR the comparison with $M$ into $f$. This block leaves the coordinate transformed; its adjoint restores it after the flag has been used.
    }
    \label{fig:circ_centres}
\end{figure}

\subsection{Weighted accumulation and thresholding}
Consider a kernel $\Lambda$. For each entry $(\eta_r,\alpha_r)$, use $Q_r:=Q^{(k,M,\eta_r)}$ to test whether $x-\eta_r$ is a centre.
If the flag $f$ is $1$ we add $\alpha_r$ to the accumulator $s$.
After the controlled addition, apply $Q_r^\dagger$ to restore $x$ and return $f$ to $0$.
We repeat that operation for all entries in the kernel.
\begin{equation}
    Z^{(k,M,\Lambda)}_{x\rightarrow s} := \prod_{r}Q^{(k,M,\eta_r)\dagger}_{x\rightarrow f} A^{\alpha_r}_{f\rightarrow s}Q^{(k,M,\eta_r)}_{x\rightarrow f}
\end{equation}
The accumulator uses signed two's-complement encoding. 
Write $W_\Lambda=\sum_r|\alpha_r|$ and choose $S$ such that $2^{S-1}>W_\Lambda$.
Every partial sum lies in $[-W_\Lambda,W_\Lambda]$, so intermediate accumulation does not overflow.
After all entries, compare the signed score to $\tau$ and XOR the answer into the geometry output register $g$. Here $L^{(\tau)}$ denotes signed comparison, implemented by flipping the most significant bit before and after an unsigned comparison with $\tau+2^{S-1}$. 
Uncompute the accumulation to clean the score register.
\begin{equation}
    G^{(k,M,\Lambda,\tau)}_{x\rightarrow g} := Z^{(k,M,\Lambda)\dagger}_{x\rightarrow s}X_gL^{(\tau)}_{s\rightarrow g}Z^{(k,M,\Lambda)}_{x\rightarrow s}
\end{equation}
The resulting oracle preserves the coordinate and leaves only the geometry bit $g$ changed, as required by \Cref{prob:geomoracle}.
We illustrate the full circuit in \Cref{fig:circ_conv}.
\begin{figure*}[!t]
    \centering
    \includegraphics[width=0.85\textwidth]{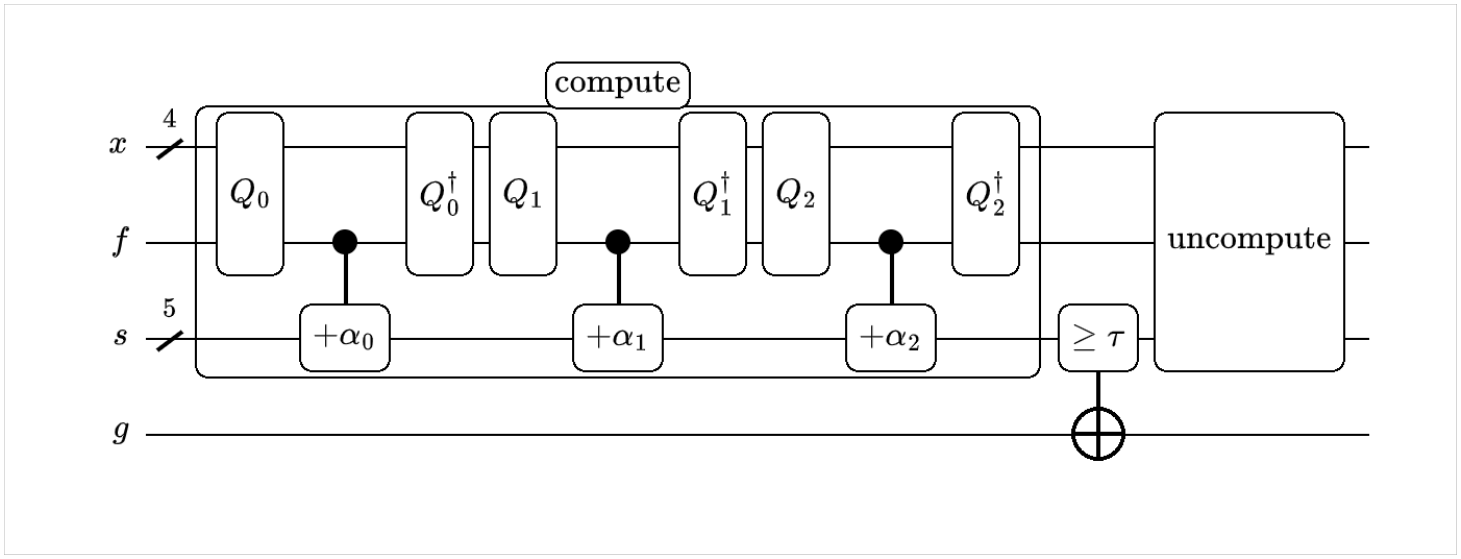}
    \caption{Local convolution oracle. For each offset $\eta_r$ we define the shifted centre query $Q_r:=Q^{(k,M,\eta_r)}$ as in \Cref{fig:circ_centres}. It supplies a flag controlling addition of $\alpha_r$ into the signed accumulator $s$. Applying $Q_r^\dagger$ after each addition restores $x$ and clears $f$. The signed score in $s$ is compared with $\tau$ and the answer is XORed into $g$. Reversing the accumulation clears $s$, leaving only $g$ changed.}
    \label{fig:circ_conv}
\end{figure*}
Unit weights and threshold one implement the fixed-shape hole rule with this accumulator circuit.

\section{Implementation and extensions}
\label{sec:implementation}
\subsection{Construct implementation}
We implemented the geometry circuits in PsiQuantum Construct\footnote{\url{https://www.psiquantum.com/construct}}.
Workbench expresses the permutation, shifted queries, and geometry routines as reusable Qubricks.
Circuit Designer displays their circuit structure, while numerical estimates and Resource Analyzer examine their cost.
We propose a companion notebook as a tutorial to follow the circuit construction on GitHub\footnote{\url{https://github.com/PsiQ/quantum_cheese}}. 
The reader can learn to generate the geometry illustrated in \Cref{fig:quantumcheese}.

\subsection{Complexity analysis of the quantum circuits}

The main parameters are the address width $n=\log_2N$, the number of permutation rounds $m$, the number of non-zero kernel entries $R$, and the accumulator width $S$.
For fixed-shape holes, all weights are one, $W_\Lambda=R$, and $S=O(\log(R+1))$.
We derive the precise cost of the oracle as a function of the chosen primitives for arithmetic operations in \Cref{app:costs}.
The primitive implementations in \Cref{subsec:primitives}~\cite{vandaele2026asymptotically-8cd,remaud2025ancilla-free-eb8} give polylogarithmic depth with a constant number of arithmetic ancillary qubits:
\begin{align}
    &\mathcal{D}_G = O(R(m\log^2(n)+\log^2(S)))\\
    &\mathcal{A}_G \leq S + 2
\end{align}
The kernel is fixed by \Cref{def:loc_material}, so $R$ and $S$ are independent of $n$.
For the numerical simulations, we used implementations with depth linear in $q$.
Both choices give the following proposition.
\begin{proposition}
    \label{prop:conv_eff}
    Consider a locally textured geometry as in \Cref{def:loc_material}, with its centre field generated by the $m$-round SPECK construction.
    Choosing $m=\Theta(n)$, the geometry oracle in \Cref{prob:geomoracle} can be implemented with polynomially-sized circuits.
\end{proposition}

\subsection{More complex geometries}
These geometry rules can be combined to describe more complex structures or multiple materials.
For example, holes of different shapes can be generated using a separate stencil and key for each shape.
Assigning different material labels to different components extends geometries beyond binary fields.

Periodic patterns, such as regular grids, also admit efficient oracles and can be combined with locally textured materials.
For example, adding pseudorandom defects to a regular grid and varying the number of centres $M$ allows one to study the effect of increasing defect density.
We discuss these extensions in more detail in Appendix~\ref{sec:extension}.

\subsection{Statistics over material distributions}
\begin{figure}
    \centering
    \includegraphics[width=0.7\linewidth]{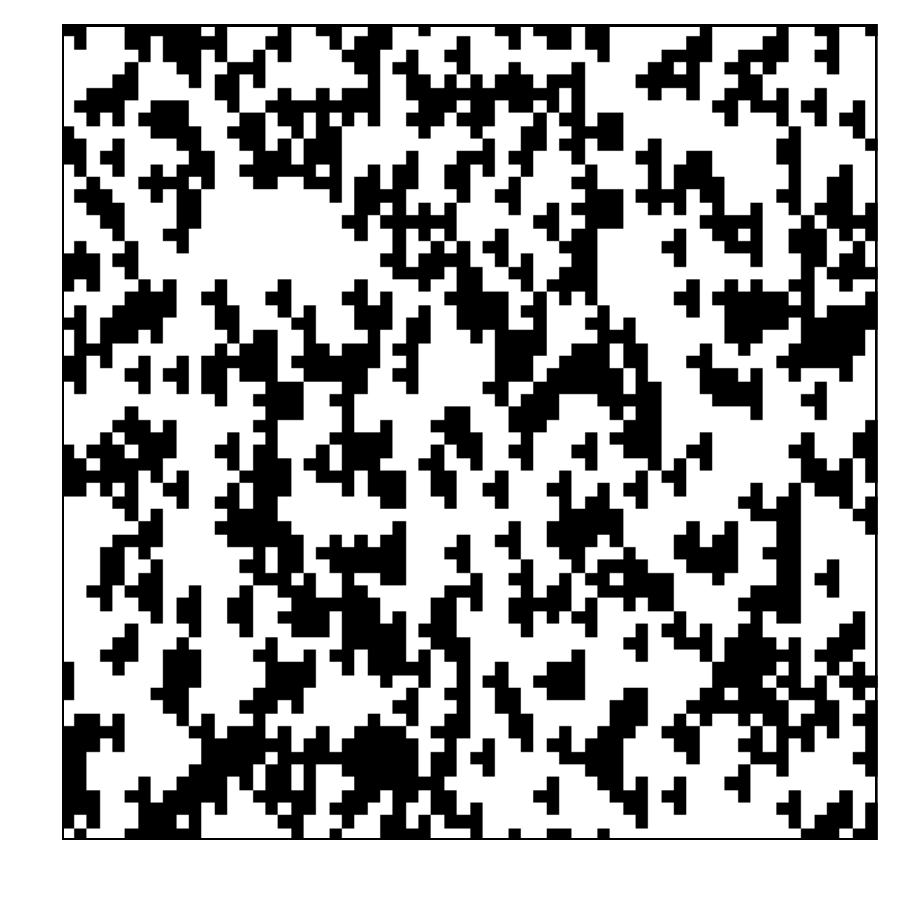}
    \caption{A $64\times64$ realization obtained by querying the geometry oracle at every coordinate in the numerical simulation of the circuit in \Cref{fig:circ_conv}. The four-pixel stencil is shown in \Cref{fig:diamond}. Parameters are $M=512$ centres, rotations $(u,v)=(1,7)$, $m=5$ rounds, and round keys $[38,51,28,14,42]$, with periodic boundaries.}
    \label{fig:quantumcheese}
\end{figure}

\label{subsec:statistics}
In this subsection we introduce a slight variation of the oracle presented before.
In the constructions of \Cref{sec:circuits}, the key of the pseudorandom function is applied from a classical register, as a wall of $X$ gates, that is, an $X$ gate is applied to the $i$-th qubit if the $i$-th bit of the key is equal to $1$.
However, we can turn the key register into a quantum register, and use it to coherently control the same $X$ gates.

Concretely, let $K$ be a key register, and let $P^{(K)}$ denote the corresponding coherent pseudorandom permutation, where each fixed XOR $X^{(k_j)}$ in the round function is replaced by a controlled XOR from the appropriate part of the key register.
For the clean centre oracle $B_{x\to f}^{(k,M)}=(P_x^{(k)})^\dagger L_{x\to f}^{(M)}P_x^{(k)}$, promoting the key to a quantum register gives
$$
    B\ket{k}_K\ket{x}\ket{0}
    =
    \ket{k}_K\ket{x}\ket{c_{k,M}(x)},
$$
where $c_{k,M}$ is the binary map associated with the key $k$.
If the key register is initialized in a superposition
$$
    \sum_k \beta_k \ket{k}_K,
$$
then the output is
$$
    \sum_k \beta_k \ket{k}_K\ket{x}\ket{c_{k,M}(x)},
$$
and similarly for the fixed-shape hole and convolution-based geometry oracles.
In other words, one gets coherent access to a superposition of material realizations rather than to a single one.
This may be useful when one is interested not only in one realization of the material, but in statistics over a whole distribution of random geometries.

\section{Outlook}
The construction gives explicit geometry access for locally textured materials with many local features as a generative rule.
Promoting the permutation keys to a quantum register gives coherent access to a distribution of material realizations.
But efficient geometry access is only one component of a quantum PDE algorithm. 
A complete application must also construct the physical operator, solve the equations, and extract useful quantities.
The next step is to combine these oracles with physical operators and PDE solvers for concrete quantities, such as flow through porous media or effective mechanical properties.

\acknowledgments
The author warmly thanks Matteo Lostaglio, Paul Mannix and Kamil Korzekwa for their precious feedback as well as PsiQuantum's Construct software team for their support.

\printbibliography

\onecolumn
\clearpage
\appendix

\section{Compiled circuit costs}
\label{app:costs}

\begin{figure}[h]
    \centering
    \includegraphics[width=\linewidth]{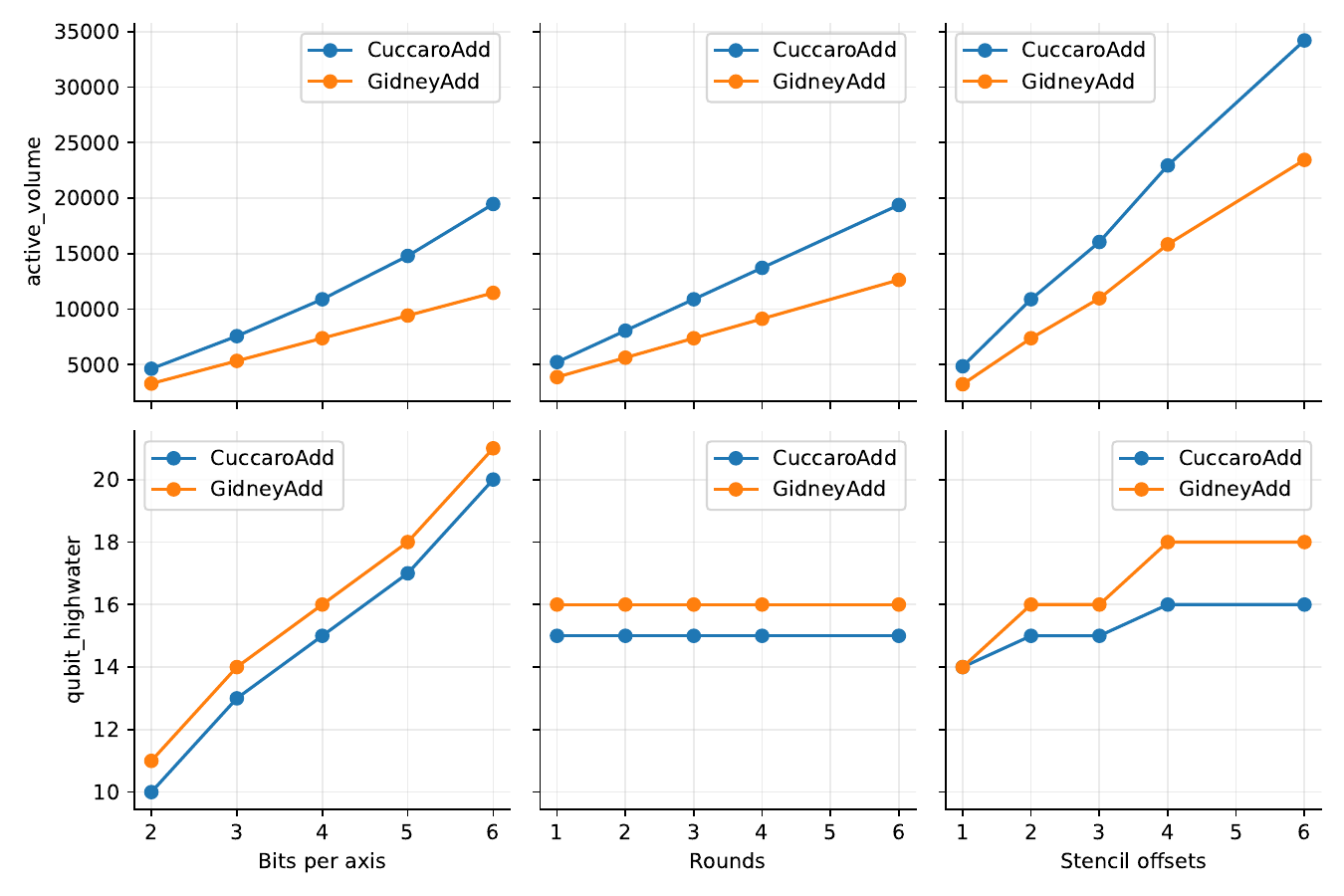}
    \caption{Resource Analyzer estimates of active volume and peak total qubit usage (qubit highwater), including input, output, and ancillary qubits. Columns vary the number of bits per axis $q$, the permutation round count $m$, and the number of stencil offsets $R$, respectively. For these square two-dimensional grids, the total address width is $n=2q$ and the number of sites is $2^n=2^{2q}$. Parameters not varied in a column are fixed at $q=4$, $m=3$, and $R=2$. The plots compare circuits using Gidney's adder~\cite{gidney2018halving-981} and Cuccaro's adder~\cite{cuccaro2004new-62d}.}
    \label{fig:sweep}
\end{figure}

Input coordinates and final geometry outputs are excluded from ancillary-qubit counts.
Arithmetic ancillary qubits are reused between primitives that restore them.
For registers of $q$ qubits we write the depth $\mathcal{D}(q)$ and ancillary-qubit count $\mathcal{A}(q)$.
The indices are $\text{q+q}$ for the quantum-quantum addition, $\text{q+c}$ for quantum-classical addition (including controlled addition) and $\text{q}<\text{c}$ for quantum-classical comparison.

One SPECK round contains an addition on $n/2$ bits and a register XOR, which has constant depth.
With $m$ rounds the permutation has depth $m(\mathcal{D}_{\text{q+q}}(n/2)+2)$ and uses $\mathcal{A}_{\text{q+q}}(n/2)$ ancillary qubits, accounting for the key XOR and register XOR.
For fixed dimension, let $\mathcal{D}_{\text{q+c}}(n)$ and $\mathcal{A}_{\text{q+c}}(n)$ also bound the depth and ancillary requirements of a component-wise coordinate shift.
For $0<M<N$, the threshold $M$ is specified by at most $n$ bits.
There is no loop over the $M$ centres, so their number does not appear as a multiplicative factor in the gate count, and the comparator costs $(\mathcal{D}_{\text{q}<\text{c}}(n),\mathcal{A}_{\text{q}<\text{c}}(n))$ in depth and ancillary qubits.
Therefore in total a single shifted query $Q_r$ costs
$$
    (\mathcal{D}_Q,\mathcal{A}_Q):=(m(\mathcal{D}_{\text{q+q}}(n/2)+2)+\mathcal{D}_{\text{q+c}}(n)+\mathcal{D}_{\text{q}<\text{c}}(n),
    \max(\mathcal{A}_{\text{q+q}}(n/2),\mathcal{A}_{\text{q+c}}(n),\mathcal{A}_{\text{q}<\text{c}}(n)))
$$

For each kernel entry, we apply $Q_r$, a controlled addition on the $S$-qubit accumulator, and $Q_r^\dagger$.
The same centre flag $f$ is reused for every entry, and is counted as an ancillary qubit, giving a cost 
$$(R(2\mathcal{D}_Q+\mathcal{D}_{\text{q+c}}(S)),1+\max(\mathcal{A}_Q,\mathcal{A}_{\text{q+c}}(S)))$$
The final signed comparison costs $(\mathcal{D}_{\text{q}<\text{c}}(S),\mathcal{A}_{\text{q}<\text{c}}(S))$, and reversing the accumulation gives the clean binary oracle.
Counting the accumulator as ancillary qubits, we obtain the bounds
$$    (\mathcal{D}_G,\mathcal{A}_G)\le(2R(2\mathcal{D}_Q+\mathcal{D}_{\text{q+c}}(S))+\mathcal{D}_{\text{q}<\text{c}}(S)+3,
    S+1+\max(\mathcal{A}_Q,\mathcal{A}_{\text{q+c}}(S),\mathcal{A}_{\text{q}<\text{c}}(S)))$$

Compiling the algorithm in PsiQuantum's Workbench gives the active-volume and peak-qubit estimates in \Cref{fig:sweep}, using the conditionally clean comparator~\cite{khattar2025rise-320} and two different adders.
Active volume counts the logical blocks needed to execute the circuit in the active-volume architecture, summed over its operations~\cite{litinski2022activevolume}. It measures the computational work in units of logical blocks, excluding idle storage.

\section{Randomness of SPECK}
\label{app:randomness}
\begin{figure}[h]
    \centering
    \includegraphics[width=0.5\linewidth]{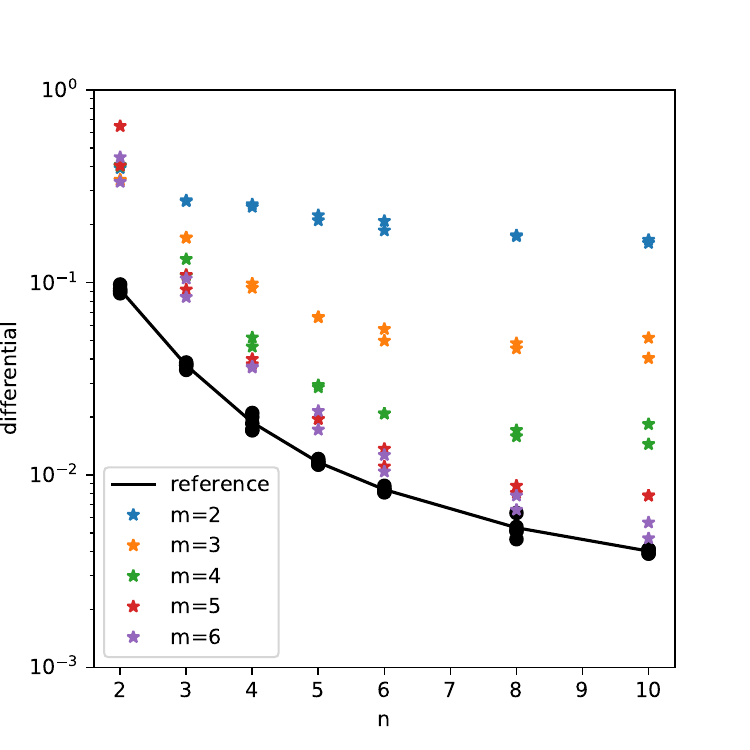}
    \caption{Differential metric $\mathcal{R}$ versus address width $n$, for different round counts $m$. Lines show averages over sampled inputs; points show individual benchmark instances.}
    \label{fig:crypto}
\end{figure}
We propose the standard differential metric to assess whether a function is random enough. 
It is computed as follows: for a fixed key $k$, sample $J$ input $n$-bit-string at random, which we call $\{x_j\}$, and compute the corresponding output $n$-bit-string $y_j = f(x_j)$.
For each input string $x_j$, flip each of the $n$ bits, for the $i$-th bit this corresponds to an XOR with $i$-th one hot encoding $e_i$.
For each bit, compute the output $y_{j,i} = f(x_j \oplus  e_i)$, and XOR this output with the original output $y_j$.
$$D_{i,j} = f(x_j \oplus  e_i) \oplus f(x_j)$$
This gives a $n$-bit offset $D_{i,j}$ that should for each bit $i$ be as spread-out as possible over the samples $j$, that is $ (D_{i,j})_{x_j \sim U(n)} \sim p_i$ should be as close to the uniform distribution as possible. 

To understand why this should be the case, consider a situation where this is not satisfied: flipping the least significant bit of the input always yields a flip of the most significant bit of the output.
In that case $p_0$ is the Dirac delta distribution centred at $e_{n-1}$.
Consider two neighbouring lattice sites, that differ only by a flip of the least significant bit, their outputs of the random function will be the same but with only the most significant bit flipped.
Therefore the output of two neighbouring lattice sites will consistently be separated by $2^{n-1}$ in their integer representation.
If we choose $M=2^{n-1}$ centres, exactly one site in each such pair is flagged after the comparison.
This is an extremely strong spatial correlation, far from the ideal uniform distribution.

To characterize how random a function is we have an approximation of $p_i$ as a histogram (thanks to the $J$ samples), and we collect the probability of the most frequent string, which we average over the $n$ bits, yielding the metric:
$$\mathcal{R} = \text{mean}_i \Big( \max_{d} (\text{counts}_j(D_{i,j} = d)/J ) \Big) $$
We show in \Cref{fig:crypto} a comparison between $m$ rounds implemented on a quantum computer, and the ideal baseline reference, which consists in outputting a random bit-string of the correct length. 
For the tested instances, increasing the number of rounds improves this diagnostic. The experiment does not establish an asymptotic mixing guarantee or uniformity over centre sets.

\section{Extensions}
\label{sec:extension}
The constructions of \Cref{sec:circuits} cover a broad family of structured random materials.
In this section, we mention a few direct extensions of the framework.

\subsection{Several shapes for holes}
In \Cref{sec:circuits} we considered sampling for only one fixed shape.
But if, for example, the goal is to model a material with circular holes of different radii, which is the model chosen in~\cite{carolis2024effect-186,roberts2000elastic-50a}, we need to be able to sample different shapes.

Use the compute architecture of \Cref{fig:circ_conv}.
For each shape $j$, let $\Delta_j$ contain $R_j$ offsets and use a separate permutation key $k_j$.
Take unit weights, accumulate the contributions from every shape into the same score register, then threshold once at $\tau=1$.
Finally reverse all contributions to clean the accumulator.
The total number of kernel entries is $R=\sum_jR_j$, and $2^{S-1}>R$ suffices for the signed accumulator.
The gate count is linear in the number of contributions, with the accumulator-width dependence given in Appendix~\ref{app:costs}.

We considered a subtractive manufacturing convention, where the domain is filled by material by default, and the geometry function describes where material is removed, as for example in cheese.
The reverse can also be considered, in an additive manufacturing convention, where the domain is empty by default and fixed-shape obstacles are added.
The comparison between the two cases is made for the modelling of ceramics in~\cite{roberts2000elastic-50a}.
More generally, one may assign different output labels to different shape families, which yields a straightforward path from mono-material to multi-material oracles.

\subsection{Periodic materials and random perturbations}
Periodic materials are not complicated to encode, but we mention them for completeness as they form the basis of an important family of procedural materials.
Consider a material with a periodic structure, that is, a repeated pattern of constant size $p_x\times p_y\in O(1)$, described by a fixed binary array
$$
    P:\{0,\dots,p_x-1\}\times \{0,\dots,p_y-1\}\to \{0,1\}.
$$
Then the geometry function, as defined in \Cref{prob:geomoracle}, on the large grid $N_x\times N_y$ is simply
$$
    g(x,y)=P(x \bmod p_x,\; y \bmod p_y).
$$
To build the oracle, one first computes the modulo of the coordinates and gets the residuals $(r_x,r_y)$, then looks up the constant-sized table $P(r_x,r_y)$, and finally uncomputes the residues to reset the ancillary qubits.
Because the period is constant, this circuit remains polynomially-sized.

Details can be found below.
We can combine this idea with the pseudorandom constructions of \Cref{sec:circuits} to implement materials that are near-periodic, that is they have a base periodic pattern that is randomly perturbed as in~\cite{klatt2020cloaking-6dd}.
For example, one can implement $M$ random bit flips on the periodic pattern, as in~\cite{rengarajan2005effect-81a}, by creating a pseudorandom binary field and XORing it with the periodic pattern.

The detailed circuit uses coordinate registers $x,y$ with widths $n_x,n_y$, output qubit $g$, and residue registers $r_x,r_y$ with a total of $\lceil\log_2p_x\rceil+\lceil\log_2p_y\rceil$ qubits.
The residues are initialized and returned to zero.
Write $x=\sum_{j=0}^{n_x-1}x_j2^j$. Then
\begin{equation}
    x\bmod p_x=\sum_{j=0}^{n_x-1}x_j(2^j\bmod p_x)\pmod{p_x}.
\end{equation}
For each bit $x_j$, conditionally add the constant $(2^j\bmod p_x)$ modulo $p_x$ into $r_x$.
Each controlled addition acts on a constant-size register and costs $O(1)$ gates, giving $O(n_x)$ gates in total.
Do the same for $y$ to compute $r_y=y\bmod p_y$ in $O(n_y)$ gates.
The constant-size lookup acts as
\begin{equation}
    U_P\ket{r_x,r_y}\ket{b}_g
    =\ket{r_x,r_y}\ket{b\oplus P(r_x,r_y)}_g.
\end{equation}
Uncompute the residues after the lookup.
The complete action is
\begin{equation}
    \ket{x,y}\ket{0,0}_{r_x,r_y}\ket{b}_g
    \longmapsto\ket{x,y}\ket{0,0}_{r_x,r_y}\ket{b\oplus g(x,y)}_g,
\end{equation}
with $O(n_x+n_y)$ gates. The construction extends component-wise to fixed dimension.

Note there is a special easy case when $p_x,p_y$ are powers of 2.
If $p_x=2^{k_x}$ and $p_y=2^{k_y}$, then the modulo operation reduces to the $k_x$ and $k_y$ lowest bits of $x$ and $y$ respectively.

\section{Lower bounds on oracles for unstructured geometries}
\label{sec:lowerbounds}
\subsection{The need for structure}
A first approach to \Cref{prob:geomoracle} for a homogeneous material with random holes could be to generate a list of holes pseudorandomly on a classical computer and use that to create a quantum oracle. 
However, for exponentially many holes, this simply would not be efficient, because the description of the features would be exponentially long.
A quantum oracle needs to have a polynomial description, although this is not sufficient.
In this section, our argument is that if the geometry encodes a hard problem, or simply lacks structure, it is impossible in general to get an efficient quantum oracle.

\subsection{Arbitrary circular holes from a database oracle}
\label{subsec:hardholes}

In this subsection we consider materials hollowed out by $M$ circular holes, whose descriptions are stored in a quantum-accessible database. 
We work in two dimensions for the sake of the example, but the same argument extends directly to balls in $d$ dimensions. 
The geometry oracle associated with a single disk is easy, given its centre and radius, membership can be computed by reversible arithmetic with polynomial overhead in the bit precision. 
However deciding whether a site lies in the union of $M$ arbitrary disks requires taking an OR over $M$ independent geometry functions. 
We show that, without additional structure, this contains unstructured search.

Consider \Cref{prob:geomoracle} for a material specified by $M$ disks, indexed by $s$, with descriptions accessible through a database oracle
\begin{equation}
    U_{\mathrm{db}}\ket{s}\ket{0}
    =
    \ket{s}\ket{x_s,y_s,r_s},
\end{equation}
where $(x_s,y_s)$ is the centre of disk $s$, and $r_s$ is its radius. 
The geometry function for disk $s$ is
\begin{equation}
    g_s(x,y)
    =
    \mathbf{1}\!\left[(x-x_s)^2+(y-y_s)^2\le r_s^2\right].
\end{equation}
Given $(x_s,y_s,r_s)$, the corresponding reversible circuit
\begin{align}
    V_{\mathrm{disk}}
    &\ket{x,y}_{xy}\ket{x_s,y_s,r_s}_{c}\ket{0}_b
    =\\
    &\ket{x,y}_{xy}\ket{x_s,y_s,r_s}_{c}\ket{g_s(x,y)}_b
    \label{eq:oracle_disk}
\end{align}
has polynomial-size implementation, since it only performs subtraction, squaring, addition, and comparison.
The full geometry function is the union-membership function
\begin{equation}
    g(x,y)
    =
    \bigvee_{s=1}^M g_s(x,y).
\end{equation}
We now show that evaluating this geometry is hard in the worst case. 
Let $z\in\{0,1\}^M$ be an arbitrary bit string, accessed through the bit-query oracle
\begin{equation}
    O_z\ket{s}\ket{b}
    =
    \ket{s}\ket{b\oplus z_s}.
\end{equation}
From $z$, define a database of disks by
\begin{equation}
    U_{\mathrm{db}}^{(z)}\ket{s}\ket{0}
    =
    \begin{cases}
        \ket{s}\ket{0,0,1}, & z_s=1,\\
        \ket{s}\ket{3,0,1}, & z_s=0.
    \end{cases}
\end{equation}
Thus, if $z_s=1$, disk $s$ is the unit disk centred at the origin, while if $z_s=0$, it is the unit disk centred at $(3,0)$.
These two disks are disjoint, and the fixed site $(0,0)$ belongs to the former but not the latter. 
Hence
\begin{equation}
    g_s(0,0)=z_s,
\end{equation}
and therefore
\begin{equation}
    g(0,0)
    =
    \bigvee_{s=1}^M z_s.
\end{equation}

Moreover, one query to $U_{\mathrm{db}}^{(z)}$ can be implemented using one query to $O_z$ and constant-size reversible logic, since the circuit only chooses between two hard-coded triples according to the value of $z_s$.
Consequently, any bounded-error quantum algorithm that evaluates the union geometry at $(0,0)$ using $T$ queries to $U_{\mathrm{db}}$ would give a bounded-error quantum algorithm for $\mathrm{OR}_M$ using $T$ queries to $O_z$. 
Since $\mathrm{OR}_M$ has quantum query complexity
\begin{equation}
    Q(\mathrm{OR}_M)=\Theta(\sqrt{M})
\end{equation}
by the Grover lower bound~\cite{beals2001quantum-b49}, the union-membership problem has worst-case query complexity
\begin{equation}
    \Omega(\sqrt{M}).
\end{equation}
In particular, if $M$ is exponential in the input size, arbitrary database access to the holes is not sufficient to obtain an efficient geometry oracle.

\begin{proposition}
    \label{thm:arbitrarydisks_hard}
    For a geometry composed of $M$ arbitrary circular holes specified by a quantum-accessible database, the geometry oracle as in \Cref{prob:geomoracle} has worst-case quantum query complexity $\Omega(\sqrt{M})$.
\end{proposition}
This lower bound shows that efficient geometry access cannot follow from database access alone. 
The constructions of \Cref{sec:circuits} therefore impose additional structure, the features are not arbitrary database entries, but are generated from a pseudorandom centre set together with fixed local transformations.

\subsection{Materials geometry from amplitude encoded states}
For some materials, the geometry emerges from the dynamics of the formation process. 
Porous structures, for example, can arise through self-organization mechanisms, described by partial differential equations.
Examples include phase separation in battery electrodes, which is often modelled by Cahn--Hilliard--reaction equations~\cite{guo2016li-6fe}, pattern formation in hydrogels~\cite{nabika2020pattern-0f6}, and polymer organisation~\cite{ohta1986equilibrium-5d6}.

Quantum algorithms solving partial differential equations, including reaction-diffusion processes~\cite{lockwood2025quantum-f02,kumar2025two-f0e,an2023quantum-d8d}, typically output the solution as an amplitude encoded quantum state. 
Suppose, for example, that a quantum circuit $U$ prepares the discretized formation field $a$ on a grid $\{x_k\}_{k=0}^{N-1}$, with $N=2^n$, as
\begin{equation}
    U \ket{0^n}
    =
    \frac{1}{\|a\|_2}
    \sum_{k=0}^{N-1} a(x_k)\ket{k}.
    \label{eq:Ua-main}
\end{equation}
A binary material geometry could then be obtained by thresholding this field, by declaring whether $a(x_k)$ lies above or below a prescribed level.
However, amplitude encoding does not immediately provide coherent pointwise access to this thresholded geometry. 
The question is therefore whether one can construct the geometric oracle $G$ as in \Cref{prob:geomoracle} from the state-preparation circuit $U$.

\begin{problem}
\label{prob:precedure}
    Let $a\in[-1,1]^N$ be nonzero, with known norm $\nu=\|a\|_2$ and $N=2^n$.
    Let $0<\eta\le\min\{1,\nu\}$ be known and assume
    $a_k\in[-1,-\eta]\cup[0,1]$ for every $k$.
    Given controlled black-box access to $U$ and $U^\dagger$, where
    \begin{equation}
        U\ket{0^n}=\frac1\nu\sum_{k=0}^{N-1}a_k\ket{k},
        \label{eq:Ua-def}
    \end{equation}
    construct an approximation to the sign geometry oracle
    \begin{equation}
        V\ket{k}\ket{b}=\ket{k}\ket{b\oplus\mathbf{1}[a_k<0]}.
        \label{eq:Va-goal}
    \end{equation}
    The approximation must have operator-norm error at most a fixed $\epsilon<1/4$ on the subspace with initialized workspace, including approximate return of that workspace to zero.
    Each use of $U$, $U^\dagger$, or a controlled version counts as one query.
    The controlled oracle fixes the global phase convention in \Cref{eq:Ua-def}.
\end{problem}

\paragraph{Lower bound.}
Consider the fields $a^{(+)}=(1,\ldots,1)$ and $a^{(-)}=(-1,1,\ldots,1)$ on $N$ sites.
Their normalized states differ only by $2\ket{0}/\sqrt{N}$, but their sign oracles must return different answers at coordinate $0$.
We choose preparation unitaries $U_\pm$ using the same fixed unitary to first prepare
\[
    \ket{v}=\frac{1}{\sqrt{N-1}}\sum_{j=1}^{N-1}\ket{j},
\]
then applying opposite rotations of angle $\theta=\arcsin(1/\sqrt{N})$ in the plane spanned by $\ket{v}$ and $\ket{0}$, acting as the identity outside that plane.
These unitaries prepare the required states and satisfy $\|U_+-U_-\|=2\sin\theta=2/\sqrt{N}$.
Replacing oracle calls one at a time therefore shows that the final states of any $q$-query algorithm satisfy
\[
    \bigl\|\ket{\Phi_+}-\ket{\Phi_-}\bigr\|
    \le \frac{2q}{\sqrt{N}}.
\]
The same bound holds for inverse and controlled queries.
Correct sign evaluation with probability at least $2/3$ requires a constant distinguishing advantage between the two instances, and hence a constant separation between the final states.
Consequently, $q=\Omega(\sqrt{N})$.

\begin{proposition}
    \label{thm:equadiffmat_hard}
    The worst-case query complexity of \Cref{prob:precedure} is $\Omega(\sqrt{N})$, even for $\eta=1$.
\end{proposition}

\end{document}